\documentclass[11pt]{amsart}
\usepackage{amssymb,color}
\numberwithin{equation}{section} 

\usepackage{graphicx}

\newcommand{\bea}{\begin{eqnarray}}
\newcommand{\eea}{\end{eqnarray}}
\newcommand{\ba}{\begin{array}}
\newcommand{\ea}{\end{array}}
\newcommand{\edc}{\end{document}}
\newcommand{\bc}{\begin{center}}
\newcommand{\ec}{\end{center}}
\newcommand{\be}{\begin{equation}}
\newcommand{\ee}{\end{equation}}

\def\ca{{\mathcal A}}

\def\cn{{\mathcal N}}

\def\bc{{\mathbb C}}

\def\bn{{\mathbb N}}

\def\br{{\mathbb R}}

\def\b{\beta}

\def\G{\Gamma}

\def\m{\mu}

\def\s{\sigma}

\def\w{\omega}
\def\Om{\Omega}

\def\h{{\mathbf{h}}}

\newtheorem{thm}{Theorem}[section]
\newtheorem{lem}[thm]{Lemma}
\newtheorem{cor}[thm]{Corollary}
\newtheorem{prop}[thm]{Proposition}

\theoremstyle{remark}
\newtheorem{rem}{Remark}[section]

\date{\today}
\begin{document}

\title[Gibbs measures]
{Gibbs measures and free energies of Ising-Vannimenus Model on the
Cayley tree}
\author{Farrukh Mukhamedov}
\address{Farrukh Mukhamedov\\
 Department of Mathematical Sciences\\
College of Science, The United Arab Emirates University\\
P.O. Box, 15551, Al Ain\\
Abu Dhabi, UAE} \email{{\tt far75m@gmail.com} {\tt
farrukh.m@uaeu.ac.ae}}

\author{Hasan Ak\i n}
\address{Hasan Ak\i n, Ceyhun Atuf Caddesi 1164, Sokak 9/4
\c{C}ankaya, TR06460, Ankara, Turkey} \email{{\tt
akinhasan25@gmail.com}}

\author{ Otabek Khakimov}
\address{ Otabek Khakimov\\
Institute of Mathematics, 29, Do'rmon Yo'li str., 100125, Tashkent,
Uzbekistan.} \email {{\tt hakimovo@mail.ru}}

\begin{abstract}
In this paper, we consider the Ising-Vannimenus model on a Cayley
tree for order two with competing nearest-neighbor and prolonged
next-nearest neighbor interactions. We stress that the mentioned
model was investigated only numerically, without rigorous
(mathematical) proofs. One of the main points of this paper is to
propose a measure-theoretical approach for the considered model. We
find certain conditions for the existence of Gibbs measures
corresponding to the model, which allowed to establish the existence
of the phase transition. Moreover, the free energies and entropies,
associated with translation invariant Gibbs measures, are
calculated. \vskip 0.3cm \noindent {\it
Mathematics Subject Classification}: 82B26, 46S10, 12J12, 39A70, 47H10, 60K35.
\end{abstract}

\maketitle

\tableofcontents
\section{Introduction}

It is known \cite{G} that the Gibbs measures are one of the
central objects of equilibrium statistical mechanics. Also, one of
the main problems of statistical physics is to describe all Gibbs
measures corresponding to the given Hamiltonian \cite{Bax}. It is
well-known that such measures form a nonempty convex compact
subset in the set of all probability measures. A simplest model in
statistical mechanics is the Ising model which has wide
theoretical interest and practical applications. There are several
papers (see \cite{NHSS1,Roz}) which are devoted to the description
of this set for the Ising model on a Cayley tree. However, a
complete result about all Gibbs measures even for the Ising model
is lacking. Later on in \cite{V} such an Ising model was
considered with next-neatest neighbor interactions on the Cayley
tree for which its phase diagram was described. The considered
model, in what follows, we will call as \textit{Ising-Vannimenus
model}. It turns out that this model has a rich enough structure
to illustrate almost every conceivable nuances of statistical
mechanics. Furthermore, intensive investigations were devoted to
generalizations of the Ising-Vannimenus model (see
\cite{Akin2017,PB,bak,FS,IT,JB,MTA,OMM,RAU} for example), but most
results of the existing works are numerically obtained. Therefore,
in \cite{NHSS} it has been proposed a measure-theoretical approach
(see \cite{G}) to the mentioned model on the Cayley tree. In fact,
the proposed approach, the authors used a dimer analogue of the
model. Therefore, that paper cannot be considered as rigorous
approach to the Ising-Vannimenus model. Hence, one of the main
aims of the present paper is to develop a measure-theoretic
approach (i.e. Gibbs measure formalism) to rigorously establish
the phase transition for the Ising-Vannimenus model on the Cayley
tree. In \cite{A1} some attempts have been made to study the phase
transition for the mentioned model within theoretical approach,
but a different definition was used for the local Gibbs measures
over the Cayley tree of order two.

On the other hand, recently, Gandolfo et al. \cite{GRRR} have
obtained some explicit formulae of the free energies (and
entropies) according to boundary conditions (b.c.) for the Ising
model on the Cayley tree. In \cite{GRS}, a wide class of new
extreme Gibbs states for the Ising model was constructed. All
these works are based on the measure-theoretic approach to study
Gibbs measures for the Ising model on the Caylay trees. Therefore,
our second aim is to present, as an illustration of the obtained
method, to find free energies of the Ising-Vannimenus model within
the developed technique.

Until now, many researchers have investigated Gibbs measures
corresponding to models with nearest-neighbor interactions on Cayley
trees. The aim of this paper is to propose rigorously the
investigation of Gibbs measures for the Ising-Vanninenus model
\cite{V} with ternary prolonged and nearest neighbor interactions on
Cayley tree.

The paper is organized as follows. In section \ref{Preliminaries},
we provide necessary notations and define the Ising-Vannimenus
model on the Caylay tree. In section \ref{Gibbs measures}, using a
rigorous measure-theoretical approach, we find certain conditions
for the existence of Gibbs measures corresponding to the model on
the Cayley tree of an arbitrary order. To describe the Gibbs
measure, we obtain a system of functional equations (which is
extremely difficult to solve). Nevertheless, in section
\ref{e-Gibbs measures}, we are able to succeed in obtaining
explicit solutions by making reasonable assumptions, for the
existence of translational invariant Gibbs measures. Furthermore,
in section \ref{phase transition}, we establish the existence of
the phase transition. We note that the periodic solutions of the
obtained functional equations will correspond to the periodic
phases of the model. We will show that to find general solutions
even in the case of translation-invariant solutions of the system
is not an easy job. The case $J_p<0$ is not easy and will require
a lot of effort to explicitly find periodic solutions. Note that
even for the usual Ising model (anti ferromagnetic case) up to now
not all periodic solutions have been found(see \cite{Roz} for the
review). In section \ref{Free energy}, as an illustration of the
developed technique, one finds the free energy associated with
translation-invariant Gibbs measures, which allows to calculate
the corresponding entropy. Finally, section \ref{Conclusions}
contains concluding remarks and discussion of the consequences of
the results.

\section{Preliminaries}\label{Preliminaries}

\subsection{Ising-Vannimenus Model with competing interactions on Cayley tree}

A Cayley tree $\Gamma^k$ of order $k\geq 1 $ is an infinite tree,
i.e., a graph without cycles with exactly $ k+1 $ edges issuing from
each vertex. Let $\Gamma^k=(V, \Lambda)$, where $V$ is the set of
vertices of $ \Gamma^k$, $\Lambda$ is the set of edges. Two vertices
$x$ and $y$ ( $x,y \in V$) are called {\it nearest-neighbors} if
there exists an edge $l\in\Lambda$ connecting them, which is denoted
by $l=<x,y>$. The distance $d(x,y), x,y\in V$, on the Cayley tree
$\Gamma^k$, is the number of edges in the shortest path from $x$ to
$y$.

For a fixed $x^0\in V$, called the root, we set
$$
W_n=\{x\in V| d(x,x^0)=n\},  \\ V_n=\{x\in V| d(x,x^0)\leq n\}
$$
and $L_n$ denotes the set of edges in $V_n$. The fixed vertex $x^0$
is called the $0$-th level and the vertices in $W_n$ are called the
$n$-th level and for $x\in W_n$ let
$$
S(x)=\{y\in W_{n}| d(x,y)=1\},
$$
be the set of direct successors of $x\in W_{n-1}$. For the sake of
simplicity we put $|x| = d(x,x^0)$, $x\in V$. Two vertices $x,y\in
V$ are called {\it the next-nearest-neighbors} if $d(x,y)=2$.
Next-nearest-neighbor vertices $x$ and $y$ are called {\it prolonged
next-nearest-neighbors} if they belong to the same branch, i.e.
$|x|\neq |y|$, which is denoted by $\widetilde{>x,y<}$.

Let spin variables $\sigma(x), x\in V,$ take values $\pm 1$. The
\textit{Ising-Vannimenus model} with competing nearest-neighbors and
next-nearest-neighbors binary interactions is defined by the
following Hamiltonian
\begin{equation}\label{hm}
H(\sigma)=-J_p\sum_{\widetilde{>x,y<}}\sigma(x)\sigma(y) - J
\sum_{<x,y>}\sigma(x)\sigma(y),
\end{equation}
where the sum in the first term ranges all prolonged
next-nearest-neighbors and  the sum in the second   term ranges
all nearest-neighbors. Here $J_p,J\in {R}$ are coupling constants.

Note that in \cite{V} it is assumed that $J>0$ and $J_p<0.$ Below
we consider the model \eqref{hm} with arbitrary sign of the
coupling constants.

As usual, one can introduce the notions of Gibbs distribution of
this model, limiting Gibbs distribution, pure phase ( extreme Gibbs
distribution ), etc (see \cite{G}, \cite{P}).

\section{Gibbs measures of the Ising-Vannimenus model}\label{Gibbs measures}

In this section we define a notion of Gibbs measure corresponding to
the Ising-Vannimenus model on an arbitrary order Cayley tree. We
propose a new kind of construction of Gibbs measures corresponding
to the model.

Below, for the sake of simplicity, we will consider a semi-infinite
Cayley tree $\Gamma_+^k$ of order $k$, i.e. an infinite graph
without cycles with $k+1$ edges issuing from each vertex except for
$x^0$ which has only $k$ edges.

We consider the model where the spin takes values in the set
$\Phi=\{-1,+1\}$ ($\Phi$ is called a {\it state space}) and is
assigned to the vertices of the tree $\G^k_+=(V,\Lambda)$. A
configuration $\s$ on $V$ is defined as a function $x\in
V\to\s(x)\in\Phi$; in a similar manner one defines configurations
$\s_n$ and $\w$ on $V_n$ and $W_n$, respectively. The set of all
configurations on $V$ (resp. $V_n$, $W_n$) coincides with
$\Omega=\Phi^{V}$ (resp. $\Omega_{V_n}=\Phi^{V_n},\ \
\Omega_{W_n}=\Phi^{W_n}$). One can see that
$\Om_{V_n}=\Om_{V_{n-1}}\times\Om_{W_n}$. Using this, for given
configurations $\s_{n-1}\in\Om_{V_{n-1}}$ and $\w\in\Om_{W_{n}}$ we
define their concatenations  by
$$
(\s_{n-1}\vee\w)(x)= \left\{
\begin{array}{ll}
\s_{n-1}(x), \ \ \textrm{if} \ \  x\in V_{n-1},\\
\w(x), \ \ \ \ \ \ \textrm{if} \ \ x\in W_n.\\
\end{array}
\right.
$$
It is clear that $\s_{n-1}\vee\w\in \Om_{V_n}$.

In the Ising-Vannimenus model spin takes values in $\Phi=\{-1,+1\}$
and the relevant Hamiltonian has the form
\begin{equation}\label{ham}
H(\sigma)=
-J_p\sum\limits_{ \widetilde{>x,y<}}\sigma(x)\sigma(y) -J
\sum\limits_{<x,y>}\sigma(x)\sigma(y),
\end{equation}

Assume that  $\h: (V\setminus\{x^0\})\times
(V\setminus\{x^0\})\times \Phi\times\Phi\to \mathbb{R}$ is a
mapping, i.e.
$$\h_{xy,uv}=(h_{xy,++},h_{xy,+-},h_{xy,-+},h_{xy,--}),$$
where $h_{xy,uv}\in \mathbb{R}$, $u,v\in\Phi$, and $x,y\in
V\setminus\{x^{(0)}\}$.

Now, we define the Gibbs measure with memory of length 2 on the
Cayley tree as follows:
\begin{equation}\label{mu}
\m^{(n)}_{\h}(\s)=\frac{1}{Z_{n}}\exp[-\beta H_n(\s)+\sum_{x\in
W_{n-1}}\sum_{y\in
S(x)}\sigma(x)\sigma(y)\h_{xy,\sigma(x)\sigma(y)}].
\end{equation}

Here, $\beta=\frac{1}{kT}$, $\sigma\in\Omega_{V_n}$ and $Z_n$ is the
corresponding to partition function
\begin{equation}\label{Zn}
Z_{n}=\sum\limits_{\sigma_n\in \Omega_{V_n}}\exp[-\beta
H(\s_n)+\sum_{x\in W_{n-1}}\sum_{y\in
S(x)}\sigma(x)\sigma(y)\h_{xy,\sigma(x)\sigma(y)}].
\end{equation}

\begin{rem} We stress that the local Gibbs measures considered in
\cite{A1} were defined as follows:
\begin{equation*}
\tilde\m^{(n)}_{\h}(\s)=\frac{1}{Z_{n}}\exp[-\beta
H_n(\s)+\sum_{x\in W_{n-1}}\sum_{y,z\in
S(x)}\sigma(x)\sigma(y)\sigma(z)\h_{xyz,\sigma(x)\sigma(y)\sigma(z)}]
\end{equation*}
which are different from \eqref{mu}.
\end{rem}

In this paper, we are interested in a construction of an infinite
volume distribution with given finite-dimensional distributions.
More exactly, we would like to find a probability measure $\m$ on
$\Om$ with given conditional probabilities $\m_{\h}^{(n)}$, i.e.
\begin{equation}\label{CM}
\m(\s\in\Om: \s|_{V_n}=\s_n)=\m^{(n)}_{\h}(\s_n), \ \ \
\textrm{for all} \ \ \s_n\in\Om_{V_n}, \ n\in\bn.
\end{equation}
If the measures $\{\m^{(n)}_{\h}\}$ are \textit{compatible}, i.e.
\begin{equation}\label{comp}
\sum_{\w\in\Om_{W_n}}\m^{(n)}_{\h}(\s\vee\w)=\m^{(n-1)}_{\h}(\s), \
\ \ \textrm{for any} \ \ \s\in\Om_{V_{n-1}},
\end{equation}
then according to the Kolmogorov's theorem there exists a unique
measure $\m_{\h}$ defined on $\Om$ with a required condition
\eqref{CM}. Such a measure $\m_{\h}$ is said to be {Gibbs measure}
corresponding to the model. Note that a general theory of Gibbs
measures has been developed in \cite{G,Roz}.

In the sequel, we need the following auxiliary fact.
\begin{lem}\label{lemma1}
If $\frac{\mathbf{a}}{\mathbf{b}}=\frac{N_1}{N_2}$,
$\frac{\mathbf{a}}{\mathbf{c}}=\frac{N_1}{N_3}$ and $\frac{\mathbf{a}}{\mathbf{d}}=\frac{N_1}{N_4}$,
then there exists $D\in \mathbb{R}$ such that  $\mathbf{a}=D N_1$,  $\mathbf{b}=D N_2$,  $\mathbf{c}=D N_3$ and  $\mathbf{d}=D N_4$.
\end{lem}

The next statement describes the conditions on the boundary fields
$\h$ guaranteeing the compatibility of the distributions
$\{\m^{(n)}_\h\}$ .

\begin{thm}\label{theorem1}
The measures  $\m^{(n)}_\h$, $n=1,2,...,$ in \eqref{mu} are
compatible iff for any $x,y\in V$ the following equations hold:
\begin{equation}\label{necessary}
\left\{
\begin{array}{ll}
e^{h_{xy,++}+h_{xy,-+}}=\prod\limits_{z\in
S(y)}\frac{\exp[\h_{yz,++}](ab)^{2}+\exp[-\h_{yz,+-}]}
{\exp[\h_{yz,++}]a^2 +\exp[-\h_{yz,+-}] b^2}\\[2mm]
e^{h_{xy,--}+h_{xy,+-}}=\prod\limits_{z\in
S(y)}\frac{\exp[-\h_{yz,-+}]+\exp[\h_{yz,--}](ab)^2}{\exp[-\h_{yz,-+}]b^2+\exp[\h_{yz,--}]a^2}\\[2mm]
e^{h_{xy,++}+h_{xy,+-}}=\prod\limits_{z\in
S(y)}\frac{\exp[\h_{yz,++}](ab)^2+\exp[\h_{yz,+-}]]}{\exp[-\h_{yz,-+}]b^2+\exp[\h_{yz,--}]a^2},\\
\end{array}
\right.
\end{equation}
where $a=\exp(\beta J)$ and $b=\exp(\beta J_p)$.
\end{thm}
\begin{proof}
{\sc Necessity}. From \eqref{comp}, we have
\begin{eqnarray}\label{comp111}
&&L_n\sum\limits_{\eta\in \Omega_{W_{n}}}\exp[-\beta H_n(\s\vee
\eta)+
\sum\limits_{x\in W_{n-1}}\sum\limits_{y\in
S(x)}\sigma(x)\sigma(y)\h_{xy,\sigma(x)\sigma(y)}]\nonumber\\[2mm]
&=&\exp [-\beta H_n(\s)+\sum\limits_{x\in W_{n-2}}\sum\limits_{y\in
S(x)}\sigma(x)\sigma(y)\h_{xy,\sigma(x)\sigma(y)}],
\end{eqnarray}
where $L_n=\frac{Z_{n-1}}{Z_n}$.

For $\s\in V_{n-1}$ and $\eta \in W_{n}$, we rewrite the Hamiltonian
as follows:

\begin{eqnarray}\label{ham1}
H_n(\s\vee\eta)
&=&-J\sum\limits_{<x,y>\in V_{n-1}}\sigma(x)\sigma(y) -J
\sum\limits_{x\in W_{n-1}}\sum\limits_{y\in
S(x)}\sigma(x)\eta(y)\nonumber \\
\nonumber
&&-J_p\sum\limits_{>x,y<\in V_{n-1}}\sigma(x)\sigma(y) -J_p
\sum\limits_{x\in
W_{n-2}}\sum\limits_{z\in S^2(x)}\sigma(x)\eta(z)\\
&=&H_n(\s_{n-1})-J\sum\limits_{x\in W_{n-1}}\sum\limits_{y\in
S(x)}\sigma(x)\eta(y)-J_p \sum\limits_{x\in
W_{n-2}}\sum\limits_{z\in S^2(x)}\sigma(x)\eta(z).
\end{eqnarray}

Therefore, the last equality with \eqref{comp111} implies
\begin{eqnarray}\label{Kolmogorov1}
&&L_n\sum\limits_{\eta\in \Omega_{W_{n}}}\exp[-\beta
H_n(\s_{n-1})-\beta J\sum\limits_{x\in W_{n-1}}\sum\limits_{y\in
S(x)}\sigma(x)\eta(y)\nonumber\\
\nonumber &-&\beta
J_p\sum\limits_{x\in W_{n-2}}\sum\limits_{z\in
S^2(x)}\sigma(x)\eta(z)+ \sum\limits_{x\in W_{n-1}}\sum\limits_{y\in
S(x)}\sigma(x)\sigma(y)\h_{xy,\sigma(x)\sigma(y)}]\\
&=&\exp [-\beta H_n(\s_{n-1})+\sum\limits_{x\in
W_{n-2}}\sum\limits_{y\in
S(x)}\sigma(x)\sigma(y)\h_{xy,\sigma(x)\sigma(y)}],
\end{eqnarray}
%
Hence, one gets
\begin{eqnarray*}\label{Kolmogorov2}
&&L_n\prod\limits_{x\in W_{n-2}}\prod\limits_{y\in S(x)}\prod\limits_{z\in S(y)}\sum\limits_{\eta(z)\in \{\mp 1\}}\exp[\sigma(y)\eta(z)\h_{yz,\sigma(y)\eta(z)}+\beta \eta(z)(J\sigma(y)+ J_p\sigma(x))]\\\nonumber
&=&\prod\limits_{x\in W_{n-2}}\prod\limits_{y\in S(x)}\exp [\sigma(x)\sigma(y)\h_{xy,\sigma(x)\sigma(y)}].
\end{eqnarray*}

Let us fix $<x,y>$. Then considering all values of $\sigma(x),
\sigma(y)\in \{-1,+1\}$, from \eqref{Kolmogorov1}, we obtain

\begin{equation}\label{Kolmogorov3}
e^{h_{xy,++}+h_{xy,-+}}=\prod\limits_{z\in S(y)}\frac{\sum\limits_{\eta(z)\in \{\mp 1\}}\exp[\eta(z)(\h_{yz,+\eta(z)}+\beta (J+ J_p))]}{\sum\limits_{\eta(z)\in \{\mp 1\}}\exp[\eta(z)(\h_{yz,+\eta(z)}+\beta (J-J_p))]}
\end{equation}
\begin{equation}\label{Kolmogorov4}
e^{h_{xy,--}+h_{xy,+-}}=\prod\limits_{z\in S(y)}\frac{\sum\limits_{\eta(z)\in \{\mp 1\}}\exp[-\eta(z)(\h_{yz,-\eta(z)}+\beta(J+ J_p))]}{\sum\limits_{\eta(z)\in \{\mp 1\}}\exp[-\eta(z)(\h_{yz,-\eta(z)}-\beta(-J+J_p))]}
\end{equation}

\begin{equation}\label{Kolmogorov4a}
e^{h_{xy,++}+h_{xy,+-}}=\prod\limits_{z\in S(y)}\frac{\sum\limits_{\eta(z)\in \{\mp 1\}}\exp[\eta(z)(\h_{yz,+\eta(z)}+\beta (J+ J_p))]}{\sum\limits_{\eta(z)\in \{\mp 1\}}\exp[-\eta(z)(\h_{yz,-\eta(z)}-\beta(-J+J_p))]}
\end{equation}
These equations imply the desired ones.\\

{\sc Sufficiency}. Now we assume that the system of equations
\eqref{necessary} is valid, then from Lemma \ref{lemma1} one finds

$$e^{\sigma(x)\sigma(y)h_{xy,\sigma(x)\sigma(y)}}D(x,y)=\prod\limits_{z\in S(y)}\sum\limits_{\eta(z)\in \{\mp 1\}}
\exp[\sigma(y)\eta(z)\h_{yz,\sigma(y)\eta(z)}+\beta
\eta(z)(J\sigma(y)+J_p \sigma(x))],$$ for some constant $D(x,y)$
depending on $x$ and $y$.

 From the last equality, we obtain
\begin{eqnarray}\label{Kolmogorov5}
&&\prod\limits_{x\in W_{n-2}}\prod\limits_{y\in S(x)}D(x,y)e^{\sigma(x)\sigma(y)h_{xy,\sigma(x)\sigma(y)}}\\\nonumber
&=&\prod\limits_{x\in W_{n-2}}\prod\limits_{y\in S(x)}\prod\limits_{z\in S(y)}\sum\limits_{\eta(z)\in \{\mp 1\}}e^{[\sigma(y)\eta(z)\h_{yz,\sigma(y)\eta(z)}+\beta \eta(z)(J\sigma(y)+ J_p \sigma(x))]}.
\end{eqnarray}
Multiply both sides of the equation \eqref{Kolmogorov5} by $e^{-\beta H_{n-1}(\sigma)}$ and denoting

$$U_{n-1}=\prod\limits_{x\in W_{n-2}}\prod\limits_{y\in S(x)}D(x,y),$$
from \eqref{Kolmogorov5}, one has

\begin{eqnarray*}
&&U_{n-1}e^{-\beta H_{n-1}(\sigma)+\sum\limits_{x\in W_{n-2}}
\sum\limits_{y\in S(x)}\sigma(x)\sigma(y)h_{xy,\sigma(x)\sigma(y)}}\\
&=&\prod\limits_{x\in W_{n-2}}\prod\limits_{y\in S(x)}
\prod\limits_{z\in S(y)}e^{-\beta H_{n-1}(\sigma)}
\sum\limits_{\eta(z)\in \{\mp 1\}}e^{[\sigma(y)\eta(z)\h_{yz,\sigma(y)\eta(z)}+\beta \eta(z)(J\sigma(y)+ J_p \sigma(x))]}.
\end{eqnarray*}
which yields
$$
U_{n-1}Z_{n-1}\m^{(n-1)}_\h(\sigma)=\sum\limits_{\eta}e^{-\beta
H_{n}(\sigma\vee\eta) +\sum\limits_{x\in W_{n-2}}\sum\limits_{y\in
S(x)}\sigma(x)\sigma(y)h_{xy,\sigma(x)\sigma(y)}}.
$$
This means
\begin{eqnarray}\label{eq4}
U_{n-1}Z_{n-1}\m^{(n-1)}_\h(\sigma)=Z_{n}\sum\limits_{\eta}\m^{(n)}_\h(\sigma\vee\eta).
\end{eqnarray}
As $\m^{(n)}_\h$ ($n\geq 1$) is a probability  measure, i.e.
$$
\sum\limits_{\sigma\in \{-1,+1\}^{V_{n-1}}}\m^{(n-1)}_\h(\sigma)
=\sum\limits_{\sigma\in \{-1,+1\}^{V_{n-1}}}\sum\limits_{\eta \in
\{-1,+1\}^{W_{n}}}\m^{(n)}_\h(\sigma\vee\eta)=1.
$$
From these equalities and \eqref{eq4} we have
$Z_{n}=U_{n-1}Z_{n-1}$. This with \eqref{eq4} implies that
\eqref{comp} holds. The proof is complete.
\end{proof}
According to Theorem \ref{theorem1} the problem of describing the
Gibbs measures is reduced to the descriptions of the solutions of
the functional equations \eqref{necessary}.

\begin{cor}\label{compatibility}
The measures $\m^{(n)}_\h$, $ n=1,2,\dots$ satisfy the compatibility
condition \eqref{comp} if and only if for any $n\in \bn$ the
following equation holds:
\begin{equation}\label{canonic_u}
\left\{\begin{array}{ll}
u_{xy,1}=a\prod\limits_{z\in S(y)}\frac{bu_{yz,3}+1}{u_{yz,3}+b}\\[4mm]
u_{xy,2}=a\prod\limits_{z\in S(y)}\frac{(bu_{yz,2}+1)u_{yz,3}}{(u_{yz,3}+b)u_{yz,1}}\\[4mm]
u_{xy,3}=a\prod\limits_{z\in
S(y)}\frac{(bu_{yz,3}+1)u_{yz,1}}{(u_{yz,2}+b)u_{yz,3}}
 \end{array}\right.
\end{equation}
where, as before $a=\exp(2\b J)$, $b=\exp(2\b J_1)$, and
\begin{equation}\label{denuh}
\begin{array}{cc}
u_{xy,1}=a\cdot\exp\left(h_{xy,++}+h_{xy,-+}\right)\\
u_{xy,2}=a\cdot\exp\left(h_{xy,--}+h_{xy,-+}\right)\\
u_{xy,3}=a\cdot\exp\left(h_{xy,++}+h_{xy,+-}\right)
\end{array}
\end{equation}
\end{cor}
It is worth mentioning that there are infinitely many solutions of
the system \eqref{necessary} corresponding to each solution of the
system of equations \eqref{canonic_u}. However, we show that each
solution of the system \eqref{canonic_u} uniquely determines a
Gibbs measure. We denote by $\mu_{\bf{u}}$ the Gibbs measure
corresponding to the solution  $\bf{u}$
 of \eqref{canonic_u}.

\begin{thm}\label{ccu}
There exists a unique Gibbs measure $\mu_\mathbf{u}$ associated with
the function $\mathbf{u}=\{\mathbf{u}_{xy}, \ \langle{x,y}\rangle\in
L \}$  where $\mathbf{u}_{xy}=(u_{xy,1},u_{xy,2},u_{xy,3})$ is a
solution of the system \eqref{canonic_u}.
\end{thm}
\begin{proof}
Let $\mathbf{u}=\{\mathbf{u}_{xy}, \ \langle{x,y}\rangle\in L \}$ be
a  function, where $\mathbf{u}_{xy}=(u_{xy,1},u_{xy,2},u_{xy,3})$ is
a solution of the system \eqref{canonic_u}. Then, for any
$h_{xy,++}\in\br$ a function $\mathbf{h}=\{\mathbf{h}_{xy},\
\langle{x,y}\rangle\in L\}$ defined by
$$
\mathbf{h}_{xy}=\left(h_{xy,++},\
\log\left(\frac{u_{xy,3}}{a}\right)-h_{xy,++}, \
\log\left(\frac{u_{xy,1}}{a}\right)-h_{xy,++},\
\log\left(\frac{u_{xy,2}}{u_{xy,1}}\right)+h_{xy,++}\right)
$$
is a solution of \eqref{necessary}.

Now fix $n\geq1$. Since $|W_{n-1}|=k^{n-1}$ and $|S(x)|=k$ we get
$|L_{n}\setminus L_{n-1}|=k^n$. Let $\s$ be any configuration on
$\Om_{V_n}$. Denote
\[
\begin{array}{ll}
\cn_{1,n}(\s)=\{\langle{x,y}\rangle\in{L_n\setminus{L_{n-1}}}:\ \s(x)=1,\ \s(y)=1,\ x\in W_{n-1},\ y\in S(x)\}\\
\cn_{2,n}(\s)=\{\langle{x,y}\rangle\in{L_n\setminus{L_{n-1}}}:\ \s(x)=1,\ \s(y)=-1,\ x\in W_{n-1},\ y\in S(x)\}\\
\cn_{3,n}(\s)=\{\langle{x,y}\rangle\in{L_n\setminus{L_{n-1}}}:\ \s(x)=-1,\ \s(y)=1,\ x\in W_{n-1},\ y\in S(x)\}\\
\cn_{4,n}(\s)=\{\langle{x,y}\rangle\in{L_n\setminus{L_{n-1}}}:\
\s(x)=-1,\ \s(y)=-1,\ x\in W_{n-1},\ y\in S(x)\}
\end{array}
\]

We have
$$
\prod_{x\in W_{n-1}\atop{y\in
S(x)}}\exp\left\{h_{xy,\s(x)\s(y)}\s(x)\s(y)\right\}=
\prod\limits_{\langle
x,y\rangle\in\cn_{1,n}(\s)}\exp\left\{h_{xy,++}\right\}
\prod\limits_{\langle
x,y\rangle\in\cn_{2,n}(\s)}\frac{a\cdot\exp\left\{h_{xy,++}\right\}}{u_{xy,3}}
$$
$$
\times\prod\limits_{\langle
x,y\rangle\in\cn_{3,n}(\s)}\frac{a\cdot\exp\left\{h_{xy,++}\right\}}{u_{xy,1}}
\prod\limits_{\langle
x,y\rangle\in\cn_{4,n}(\s)}\frac{u_{xy,2}\exp\left\{h_{xy,++}\right\}}{u_{xy,1}}
$$
$$
=\prod\limits_{\langle x,y\rangle\in L_n\setminus
L_{n-1}}\exp\left\{h_{xy,++}\right\}\prod\limits_{\langle
x,y\rangle\in\cn_{2,n}(\s)}\frac{a}{u_{xy,3}} \prod\limits_{\langle
x,y\rangle\in\cn_{3,n}(\s)}\frac{a}{u_{xy,1}} \prod\limits_{\langle
x,y\rangle\in\cn_{4,n}(\s)}\frac{u_{xy,2}}{u_{xy,1}}
$$
By means of the last equality, from \eqref{mu} and \eqref{Zn} we
find
$$
\mu_{\h}^{(n)}(\s)=\frac{\exp\{-\b H_n(\s)\}\prod\limits_{x\in
W_{n-1}\atop{y\in
S(x)}}\exp\left\{h_{xy,\s(x)\s(y)}\s(x)\s(y)\right\}}
{\sum\limits_{\w\in\Om_{V_n}}\exp\{-\b H_n(\w)\}\prod\limits_{x\in
W_{n-1}\atop{y\in
S(x)}}\exp\left\{h_{xy,\s(x)\w(y)}\s(x)\w(y)\right\}}
$$
\begin{equation}\label{mu_u}
=\frac{\exp\{-\b H_n(\s)\}\prod\limits_{\langle
x,y\rangle\in\cn_{2,n}(\s)}\frac{a}{u_{xy,3}} \prod\limits_{\langle
x,y\rangle\in\cn_{3,n}(\s)}\frac{a}{u_{xy,1}} \prod\limits_{\langle
x,y\rangle\in\cn_{4,n}(\s)}\frac{u_{xy,2}}{u_{xy,1}}}
{\sum\limits_{\w\in\Om_{V_n}}\exp\{-\b
H_n(\w)\}\prod\limits_{\langle
x,y\rangle\in\cn_{2,n}(\w)}\frac{a}{u_{xy,3}} \prod\limits_{\langle
x,y\rangle\in\cn_{3,n}(\w)}\frac{a}{u_{xy,1}} \prod\limits_{\langle
x,y\rangle\in\cn_{4,n}(\w)}\frac{u_{xy,2}}{u_{xy,1}}}
\end{equation}

One can see the right hand side of \eqref{mu_u} does not depend to
$h_{xy,++}$. So, we can say that each solution $\mathbf{u}$ of the
system \eqref{canonic_u} uniquely determines only one Gibbs measure
$\mu_{\mathbf{u}}$.
\end{proof}

\begin{rem}
Hence, due to Theorem \ref{ccu} there exists a phase transition for
the model \eqref{ham1} with $J_0=0$ if and only if the equation
\eqref{canonic_u} has at least two solutions.
\end{rem}

\begin{rem} We point out that in the original work \cite{V} modulated phases were found in the frustrated
regime, when the next-nearest-neighbor interaction $J_p$ is
negative. To obtain these kinds of phases, one needs to find
periodic solutions\footnote{Periodicity of the solution
$\mathbf{u}=\{\mathbf{u}_{xy}, \ \langle{x,y}\rangle\in L \}$ can
be defined via representing the tree as a free group. We refer the
reader to \cite{Roz} for detail information.} of the equation
\eqref{canonic_u}. We will show that to find general solutions
even in the case of translation-invariant ones of the system is
not an easy job. Our main aim in this paper is first rigorously
establish the existence of the phase transition by finding
translation-invariant solutions of the system. The case $J_p<0$ is
not easy and will require a lot of effort to explicitly find
periodic solutions. Note that even for the usual Ising model (anti
ferromagnetic case), up to now, not all periodic solutions have
been found (see \cite{Roz} for the review).
\end{rem}
\section{The existence of Gibbs Measures}\label{e-Gibbs measures}

In this section we are going to establish the existence of Gibbs
measures by analyzing the equation \eqref{canonic_u}.

Recall that ${\bf u}=\{{\bf u}_{xy}\}_{\langle{x,y}\rangle\in L}$
is a translation-invariant function, if one has
$\mathbf{u}_{xy}=\mathbf{u}_{zw}$ for all
$\langle{x,y}\rangle,\langle{z,w}\rangle\in L$. A measure $\m_{\bf
u}$, corresponding to a translation-invariant function ${\bf u}$,
is called a {\it translation-invariant Gibbs measure}.

Solving the equation \eqref{canonic_u}, in general, is rather very
complex. Therefore, let us first restrict ourselves to the
description of its translation-invariant solutions. Hence,
\eqref{canonic_u} reduces to the following one
\begin{equation}\label{tru}
\left\{\begin{array}{ll}
u_1=a\left(\frac{bu_3+1}{u_3+b}\right)^k\\[3mm]
u_2=a\left(\frac{(bu_2+1)u_3}{(u_3+b)u_1}\right)^k\\[3mm]
u_3=a\left(\frac{(bu_3+1)u_1}{(u_2+b)u_3}\right)^k
\end{array}\right.
\end{equation}

\subsection{Solution of the system \eqref{tru}}

In this subsection, we are aiming to study the set of all solutions
of the system \eqref{tru}.

Denote $\sqrt[k]{u_1}=x_1,\ \sqrt[k]{u_2}=x_2,\ \sqrt[k]{u_3}=x_3$,
$\tilde{a}=\sqrt[k]{a}$. Then from \eqref{tru} we obtain
\begin{equation}\label{xyz}
\left\{\begin{array}{lll}
x_1=\tilde{a}\frac{bx_3^k+1}{x_3^k+b}\\[3mm]
x_2=\tilde{a}\frac{(bx_2^k+1)x_3^k}{(x_3^k+b)x_1^k}\\[3mm]
x_3=\tilde{a}\frac{(bx_3^k+1)x_1^k}{(x_2^k+b)x_3^k}
\end{array}\right.
\end{equation}
Define the following sets
$$
\begin{array}{lll}
\ca_1=\left\{{\bf{x}}\in\br^3_+: x_1=x_2\right\}, & &
\ca_2=\left\{{\bf{x}}\in\br^3_+: x_1=x_3\right\}\\
\ca_3=\left\{{\bf{x}}\in\br^3_+: x_2=x_3\right\}, & &
\ca=\left\{{\bf{x}}\in\br^3_+: x_1=x_2=x_3\right\}
\end{array}
$$
\begin{prop}\label{pr1}
Let $\bf{x}$ be a solution of \eqref{xyz}. Then $\bf{x}\in\ca$ if
and only if ${\bf{x}}\in\ca_1\cup\ca_2\cup\ca_3$.
\end{prop}
\begin{proof}
Assume that $x_1=x_2$, then multiplying the second and the third
equalities of \eqref{xyz} and dividing the obtained one by the first
equality, one finds
$$
\frac{x_2x_3}{x_1}=\tilde{a}\frac{bx_2^k+1}{x_2^k+b}.
$$
The assumption yields $x_3=\tilde{a}\frac{bx_1^k+1}{x_1^k+b}$. Now
inserting the last equality into the third one of \eqref{xyz}, we
obtain $x_1=x_3$.

Let us suppose $x_1=x_3$, then from the first and the third
equations of \eqref{xyz}, one gets $x_1^k+b=x_2^k+b$, which implies
$x_1=x_2$.

Now assume that $x_2=x_3$. In this case dividing the second by the
third equations of \eqref{xyz} we obtain
$\left(\frac{x_3}{x_1}\right)^k=1$, which yields $x_1=x_3$. This
completes the proof.
\end{proof}

A natural question arises: Does there exist a solution on $\mathbb
R^3_+\setminus\mathcal A$? We will try to answer it in the next
subsection 4.2.

Now we consider the case $u:=u_1=u_2=u_3$. Then denoting $z=bu$
from \eqref{tru} we obtain
\begin{equation}\label{UU}
\left(\frac{z+1}{z+b^2}\right)^k=\frac{b^{k-1}}{a}z
\end{equation}

To solve the last equation we apply the following well-known fact.

\begin{prop}\cite{P}\label{prest}
The equation
$$
\left(\frac{1+x}{b+x}\right)^{m-1}=ax
$$
(with $x\geq0, m\geq2, a>0, b>0$) has one solution if either $m=2$
or $b\leq\left(\frac{m}{m-2}\right)^2$. If $m>2$ and
$b>\left(\frac{m}{m-2}\right)^2$ then there exist $\eta_1(b,m),
\eta_2(b,m)$ with $0<\eta_1(b,m)<\eta_2(b,m)$ such that the equation
has three solutions if $\eta_1(b,m)<a<\eta_2(b,m)$ and has two
solution if either $a=\eta_1(b,m)$ or $a=\eta_2(b,m)$. In fact
$$
\eta_i(b,m)=\frac{1}{x_i}\left(\frac{1+x_i}{b+x_i}\right)^{m-1},
$$
where $x_1, x_2$ are solutions of
$$
x^2+[2-(b-1)(m-2)]x+b=0.
$$
\end{prop}

Hence, according to Proposition \ref{prest} the equation
\eqref{UU} can be solved under certain conditions which provide
sufficient conditions for the existence of Gibbs measures.
\subsection{Solution on  $\mathbb
R^3_+\setminus\mathcal A$}. In this subsection, for the sake of
simplicity, we assume that the order of the tree is two, i.e. $k=2$.

Suppose that ${\bf{x}}\in\br^3_+\setminus\ca$ is a solution of
\eqref{xyz}. Due to Proposition \ref{pr1} we assume that $x_1=x,\
x_2=mx,\ x_3=tx$, where $x,m,t>0,\ m\neq1,\ t\neq1$ and $m\neq t$.
It then follows from \eqref{xyz} that
\begin{equation}\label{x,mx,tx}
\left\{\begin{array}{lll}
x=\tilde{a}\frac{bt^2x^2+1}{t^2x^2+b}\\[3mm]
mx=\tilde{a}t^2\frac{bm^2x^2+1}{t^2x^2+b}\\[3mm]
t^3x=\tilde{a}\frac{bt^2x^2+1}{m^2x^2+b}
\end{array}\right.
\end{equation}
Note that $mt=1$ if and only if $m=t=1$. Indeed, multiplying the
second and the third equations of \eqref{x,mx,tx}, and dividing the
result by the first one, we get
$$
mtx=\tilde{a}\frac{bm^2x^2+1}{m^2x^2+b}
$$
The last equality with the first equation of \eqref{xyz} implies
that $m=t=1$ if $mt=1$. So, in the current setting, we may assume
that $mt\neq1$. From \eqref{x,mx,tx} one finds
\begin{equation}\label{1}
\left\{\begin{array}{ll}
b(m^2-m)t^2x^2=m-t^2\\
(m^2t-1)t^2x^2=b(1-t^3)
\end{array}\right.
\end{equation}
Since $x>0$, we get $1<t<m^{-2}$ or $m^{-2}<t<1$. Plugging \eqref{1}
into the first equation of \eqref{x,mx,tx}, one finds
\begin{equation}\label{2}
x=\tilde{a}b^{-1}\frac{m^2t-1}{t(m^2-m)}
\end{equation}
Substituting \eqref{2} into \eqref{1} yields
$$
\tilde{a}^2(m^2t-1)^2=b(m^2-m)(m-t^2)
$$
The last equality with \eqref{1} implies that $(x,mx,tx)\in\mathbb
R^3_+\setminus\mathcal A$ is a solution of \eqref{x,mx,tx} if and
only if the parameters $m$ and $t$ satisfy the following equations
$$
\left\{\begin{array}{lll}
b^2(m^2-m)(1-t^3)=(m-t^2)(m^2t-1)\\
\tilde{a}^2(m^2t-1)^2=b(m^2-m)(m-t^2)\\
1<t<m^{-2}\ \ \mbox{or}\ \ m^{-2}<t<1
\end{array}\right.
$$
where $x$ is defined by \eqref{2}.

It is easy to check that
$\left(\sqrt{b},3\sqrt{b},\frac{\sqrt{b}}{2}\right)$ is a solution
of \eqref{x,mx,tx},
 if $\tilde{a}=\frac{6}{7}\sqrt{b^3}$ and $b=\sqrt{\frac{11}{6}}$.

\section{The existence of phase transition}\label{phase transition}

In this section, we restrict ourselves to the case $k=2$ and
\begin{eqnarray}\label{req1}
h_{xy,++}= h_{xy,-+}=h_1 \mbox{ and } h_{xy,--}= h_{xy,+-}=h_2.
\end{eqnarray}

The analysis of the solution of the equations \eqref{necessary} is
rather tricky. In this section, we will study the
translation-invariant solutions.

Denoting $\ln u_1=h_{xy,++}= h_{xy,-+}$ and $\ln u_2=h_{xy,--}=
h_{xy,+-}$ for any $x,y\in V$, from \eqref{necessary} one can
produce
\begin{equation}\label{newcase1}
u_1^2=\bigg(\frac{a^2 b^2 u_1u_2+1}{a^2u_1u_2+b^2}\bigg)^2=u_2^2,
\end{equation}
where $a=e^{\b J}$ and $b=e^{\b J_p}$. This means $u_1=u_2$,
therefore letting $u:=u_1=u_2$, we have
$$
u=\frac{(a b)^2 u^2+1}{a^2 u^2+b^2}.$$ Now putting $c:=a^2$ and
$d:=b^2$, then the last equation reduces to
\begin{equation}\label{case12}
u=g(u),
\end{equation}
where
\begin{equation}\label{case11}
g(u)=\frac{cd u^2+1}{c u^2+d}.
\end{equation}

Note that if there is more than one positive solutions  of
\eqref{case12}, then we have more than one translation-invariant
Gibbs measures corresponding to the solution of \eqref{case12}.


\begin{prop}\label{Proposition-k=2}
The equation $u=\frac{c d u^2+1}{c u^2+d}$ has one solution if
either $c\leq 1$ or $d<3$. If $d\geq3$ then there exist
$\eta_1(d)$, $\eta_2(d)$ with $0<\eta_1(d)<\eta_2(d)$ such that
equation \eqref{case12} has three solutions if
$\eta_1(d)<c<\eta_2(d)$ and has two solutions if either
$\eta_1(d)=c$ or $\eta_2(d)=c$.
\end{prop}
\begin{figure} [!htbp]\label{roots}
 \centering
\includegraphics[width=70mm]{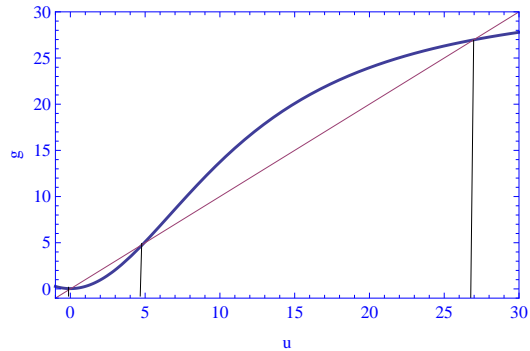}
\caption{Graph of the function $g$ defined in \eqref{case11} for
the parameters $J=-1.85,J_p=4.5,T=2.6$.}\label{roots}
\end{figure}

The proof of  Proposition \ref{Proposition-k=2} can be done by using
the similar method as in \cite{NHSS}.


From the proposition we infer that if one has $\b J_p>(\ln{3})/2$
then there are three translation-invariant Gibbs measures for the
model, which yields the existence of the phase transition.

Let us give an illustrative example.  Fig. \ref{roots} shows that
there are 3 positive fixed points of the function \eqref{case11},
if we take $J=-1.85,J_p=4.5,T=2.6$. In Fig \ref{roots}, we can
find three fixed points of the function $g$ as
$u_1=0.0316222,u_2=4.86623,u_3=26.9681$ corresponding to the
parameters $J = -1.85, J_p = 4.5, T = 2.6$. We have 3 TIGMs
associated with the fixed points
$u_1=0.0316222,u_2=4.86623,u_3=26.9681$. Therefore, the phase
transition for the model \eqref{hm} occurs.

\section{Free energy}\label{Free energy}

In this section, we study the free energy of depending on the
boundary conditions for the Ising-Vannimenus model on the Cayley
tree. From the pervious section we know that for any boundary
condition satisfying the equations \eqref{necessary} there exist
Gibbs measures corresponding to the model. Recall that the partition
function of the model is
\begin{eqnarray}\label{partition1}
Z_n=Z_n(\beta, \h)=\sum\limits_{\s\in \Omega_{V_n}}\exp \big\{-\beta
H_{n}(\s)+ \sum\limits_{x\in W_{n-1}}\sum\limits_{y\in
S(x)}\sigma(x)\sigma(y)h_{xy,\sigma(x)\sigma(y)}\big\}.
\end{eqnarray}

Then the free energy is defined as follows:
\begin{eqnarray}\label{free1}
F(\beta,\h)=\lim\limits_{n\to \infty }\frac{1}{\beta |V_n|}\ln
Z_n(\beta, \h).
\end{eqnarray}

In this section, we discuss the behavior of the free energy of the
model, as function of boundary conditions. As before, we consider
the boundary conditions \eqref{req1}, i.e.
\begin{equation}\label{req12}
h_{xy,++}= h_{xy,-+}=\tilde h_1 \mbox{ and } h_{xy,--}=
h_{xy,+-}=\tilde h_2. \ \ \ \forall <x,y>\in L.
\end{equation}

\begin{prop}\label{prop-free}
The free energies corresponding to the translation-invariant (TI)
boundary conditions with \eqref{req12} exist and are given by
\begin{eqnarray}\label{FE-TI1-case1}
F_{TI_{1}}(\beta, h_i)=-\frac{1}{\beta}\log
\big[2\cosh(h_i+\beta(J+J_p))\cosh(h_i+\beta(J-J_p))\big],
\end{eqnarray}
where $h_i$, ($i=1,2,3$), is the variety such that $u_i=e^{h_i}$
which is a solution of \eqref{case12}.
\end{prop}
\begin{proof} From \eqref{req1} one finds
\begin{eqnarray}\label{10}\nonumber
D(x,y)e^{h_{xy,++}}&=&\prod\limits_{z\in S(y)}\big[e^{h_{yz,++}+\beta(J+J_p)}+e^{-h_{yz,+-}-\beta(J+J_p)}\big]\\
&=&\prod\limits_{z\in S(y)}2e^{\frac{h_{yz,++}-h_{yz,+-}}{2}}\cosh\big[\frac{h_{yz,++}-h_{yz,+-}}{2}+\beta(J+J_p)\big].
\end{eqnarray}
\begin{eqnarray}\label{11}\nonumber
D(x,y)e^{-h_{xy,-+}}&=&\prod\limits_{z\in S(y)}\big[e^{h_{yz,++}+\beta(J-J_p)}+e^{-h_{yz,+-}-\beta(J-J_p)}\big]\\
&=&\prod\limits_{z\in S(y)}2e^{\frac{h_{yz,++}-h_{yz,+-}}{2}}\cosh\big[\frac{h_{yz,++}-h_{yz,+-}}{2}+\beta(J-J_p)\big].
\end{eqnarray}
Multiply the equations \eqref{10} and \eqref{11}, we then obtain
\begin{eqnarray}\label{12}
D(x,y)
&=&4\prod\limits_{z\in S(y)}b(y,z),
\end{eqnarray}
where
$$b(y,z)=e^{\frac{h_{yz,++}-h_{yz,+-}}{2}}\big(\cosh\big[\frac{h_{yz,++}-h_{yz,+-}}{2}+\beta(J+J_p)\big]
\cosh\big[\frac{h_{yz,++}-h_{yz,+-}}{2}+\beta(J-J_p)\big]\big)^{\frac{1}{2}}.
$$
Hence, one finds
\begin{eqnarray}\label{13}\nonumber
U_{n-1}&=&\prod\limits_{x\in W_{n-2}}\prod\limits_{y\in
S(x)}D(x,y)\\\nonumber &=&4^{|W_{n-1}|}\prod\limits_{y\in
W_{n-1}}\prod\limits_{z\in
S(y)}b(x,y)\\[2mm]\nonumber
&=&4^{|W_{n-1}|}e^{\sum\limits_{y\in W_{n-1}}\sum\limits_{z\in
S(y)}\ln b(x,y)}.
\end{eqnarray}

Denoting $ \mathfrak{a}(x,y)=\ln {b}(x,y)$, from the last equality
we get
\begin{eqnarray}\label{14}\nonumber
Z_{n}&=&U_{n-1}Z_{n-1}\\
&=&4^{|V_{n-1}|}e^{\sum\limits_{y\in W_{n-1}}\sum\limits_{z\in S(y)}\mathfrak{a}(y,z)}e^{\sum\limits_{y_1\in W_{n-2}}
\sum\limits_{z_1\in S(y)}\mathfrak{a}(y_1,z_1)}\ldots e^{\sum\limits_{\widetilde{y}\in W_{0}}
\sum\limits_{\widetilde{z}\in S(\widetilde{y})}\mathfrak{a}(\widetilde{y},\widetilde{z})}\\\nonumber
&=&4^{|V_{n-1}|}e^{\sum\limits_{<x,y>\in V_{n}}\mathfrak{a}(x,y)}.
\end{eqnarray}
%
Therefore, from \eqref{req12} we have
\begin{eqnarray}\label{FE-TI1}\nonumber
F_{TI_{1}}(\beta, \h)&=&-\lim\limits_{n\to \infty }\frac{|V_{n-1}|}{\beta |V_n|}\ln D(x,y)\\
&=&-\frac{1}{\beta}\ln\big[2e^{\tilde h_1-\tilde
h_2}\cosh(\frac{\tilde h_1+\tilde
h_2}{2}+\beta(J+J_p))\cosh(\frac{\tilde h_1+\tilde
h_2}{2}+\beta(J-J_p))\big].
\end{eqnarray}
Due to \eqref{newcase1} we may assume that $e^{\tilde h_1}=e^{\tilde
h_2}$. So, due to Proposition \ref{Proposition-k=2}, under certain
conditions, there exist three solutions $u_i$ ($i=1,2,3$) of
\eqref{case12}. Denoting $h_i=\ln u_i$ from \eqref{FE-TI1}, one
finds
\begin{eqnarray}\label{FE-TI1-case1}
F_{TI_{1}}(\beta, h_i)
&=&-\frac{1}{\beta}\ln\big[2\cosh(h_i+\beta(J+J_p))\cosh(h_i+\beta(J-J_p))\big].
\end{eqnarray}
This completes the proof.
\end{proof}
\begin{figure}[!htbp]\label{free}
 \centering
\includegraphics[width=80mm]{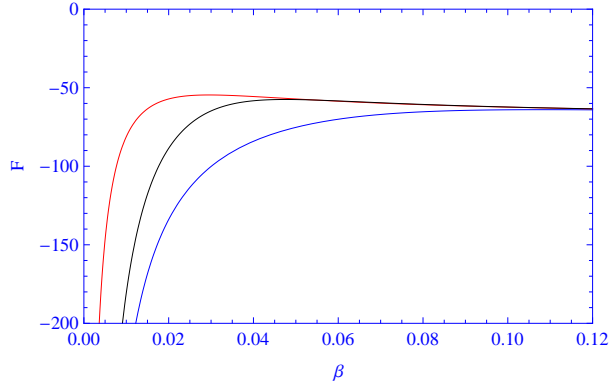}
\caption{The free energies $F_{TI_{1}}(\beta, h)$ (blue color
line, for $u_1 =0.260261$), (red color, for $u_2 = 1.18483$) and
(black color line, for $u_1 =0.260261$). Here $J_p = 4.5, J
=-1.85$.}\label{free}
\end{figure}

In order to draw the free energy $F_{TI_{1}}(\beta, h)$ as a
function of $\beta$, we consider the fixed points of the function
\eqref{case11}. Some particular plots are shown in Fig. \ref{free},
where $u_1 =0.260261$, $u_2 = 1.18483$ and $u_1 =0.260261$ are the
fixed points of the function $g$ corresponding to the parameters
$J_p = 4.5, J =-1.85$ and $T = 2.6$.

Let us compute the entropy
\begin{eqnarray}\label{FE-TI1-case1}
S(\beta,h_i)&=&-\frac{d F(\beta,h_i)}{dT}=\frac{d
F(\beta,h_i)}{d\beta}\frac{1}{\beta^2}\\\nonumber
&=&-\frac{-\ln\left[2 \cosh(h_i+\beta (J-J_p)) \cosh(h_i+\beta
(J+J_p))\right]}{\beta ^4}\\\nonumber
&&-\frac{\beta(J-J_p)\tanh(h_i+\beta (J-J_p))+\beta (J+J_p) \tanh
(h_i+\beta (J+J_p))}{\beta ^4}.
\end{eqnarray}

\section{Conclusions}\label{Conclusions}

 In the present paper, we have considered the
Ising-Vannimenus model on an arbitrary order Cayley tree with
competing nearest-neighbor, prolonged next-nearest neighbor
interactions. Recently, the mentioned model was investigated only
numerically, without rigorous (mathematical) proofs \cite{NHSS}.
We have proposed a measure-theoretical approach in order to study
the translation invariant Gibbs measures associated with the
model. Under certain conditions the existence of Gibbs measures of
the Ising-Vannimenus model is obtained. Then we have established
the existence of the phase transition. Moreover, an explicit
formulae of the free energies corresponding to the translation
invariant Gibbs measures is found. Also, we have calculated the
entropies corresponding to the mentioned free energies.

We point out that for the Ising-Vannimenus model, the free
energies and entropies associated with various known boundary
conditions such as ART \cite{ART}, Bleher-Ganikhodjaev \cite{BG},
Zachary \cite{Zachary},  have not been investigated, yet. Explicit
formulaes of the free energies and entropies for the mentioned
boundary conditions will be calculated in the future publications.

\section*{Acknowledgments} The authors are grateful to an anonymous
referee whose useful comments and suggestions improved the
presentation of the paper.


\begin{thebibliography}{99}


\bibitem{A1}  Ak\i n H.,
Using new approaches to obtain Gibbs measures of Vannimenus model
on a Cayley tree, \textit{Chinese Journal of Physics}, {\bf 54}
635-649 (2016).

\bibitem{Akin2017}  Ak\i n H.,
Phase transition and Gibbs Measures of Vannimenus model on
semi-infinite Cayley tree of order three, \emph{Int. J. Mod. Phys.
B}, \textbf{31}, 1750093 (2017) [17 pages] 



\bibitem{ART} Ak\i n H., Rozikov U.A. and Temir S.,
A new set of limiting Gibbs measures for the Ising model on a Cayley tree.
\textit{J. Stat. Phys.} \textbf{142}, 314-321 (2011).

\bibitem{BG}  Bleher P.M. and Ganikhodjaev N.N.,
On pure phases of the Ising model on the Bethe lattice, Theor.
Probab. Appl. 35, 216-227 (1990).
%


\bibitem{PB} Bak P.,
Chaotic Behavior and Incommensurate Phases in the Anisotropic Ising Model with Competing Interactions,
{\it Phys.Rev. Lett.} {\bf 46}, 791-794 (1981)

\bibitem{bak} Bak P.,
Commensurate phases, incommensurate phases and the devil's
staircase, {\it Rep. Prog. Phys.} {\bf 45}, 587-629 (1982)


\bibitem{Bax} Baxter R.J., {\it Exactly Solved Models in Statistical  Mechanics}, Academic Press, London/ New York, (1982)

%

%


\bibitem{FS} Fisher M.E. and Selke W., Infinitely Many Commensurate Phases in a Simple Ising Model,
{\it Phys.Rev.Lett.} {\bf 44}, 1502-1505   (1980)




\bibitem{NHSS1} Ganikhodjaev N., Ak\i n H., Uguz S. and  Temir S.,
Phase diagram and extreme Gibbs measures of the Ising model on a
Cayley tree in the presence of competing binary and ternary
interactions, \textit{Phase Transitions} \textbf{84} (2011),
1045--1063.

\bibitem{NHSS} Ganikhodjaev N., Ak\i n H, Uguz S. and  Temir S.,
On extreme Gibbs measures of the Vannimenus model, \textit{J.
Stat. Mech.} (2011) P03025.

\bibitem{GRS} Gandolfo D., Ruiz J. and Shlosman S.,
A manifold of pure Gibbs states of the Ising model on a Cayley
tree, \textit{J. Stat. Phys.} \textbf{148} 999-1005 (2012).


\bibitem{GRRR} Gandolfo D.,  Rakhmatullaev M.M., Rozikov U.A. and  Ruiz J.,
On free energies of the Ising model on the Cayley tree, \textit{J.
Stat. Phys.} \textbf{150} (6), 1201-1217 (2013).


%

\bibitem{G} Georgii H.-O., \textit{Gibbs Measures and Phase Transitions} (de Gruyter Stud. Math., Vol.9), Walter de Gruyter, Berlin, New York
(1988)

%

\bibitem{IT} Inawashiro S. and Thompson C.J.
Competing Ising Interactions and Chaotic Glass-Like Behaviour on a Cayley Tree,
{\it Physics Letters }, {\bf 97A}, 245-248 (1983).


\bibitem{JB} Jensen M.H. and Bak P.,
Mean-field theory of the three-dimensional anisotropic Ising model as a four-dimensional mapping,
{\it Phys. Rev. B} {\bf 27}, 6853-6868 (1983)




\bibitem{MTA}  Mariz M., Tsallis C. and Albuquerque A.L. Phase Diagram of the Ising Model on a Cayley tree in the Presence of
competing Interactions and Magnetic Field, {\it J. Stat. Phys.} {\bf
40}, 577-592 (1985).

%

\bibitem{OMM} Ostilli M., Mukhamedov F. and Mendes J.F.F. Phase diagram of an Ising model with competitive interactions on a Husimi tree and its disordered
counterpart, \textit{Physica A}, {\bf 387} (2008) 2777--2792.


\bibitem{P} Preston Ch. J., \textit{Gibbs States on Countable Sets}, Cambridge Univ.Press, Cambridge (1974).

\bibitem{Roz} Rozikov U.A., \textit{Gibbs Measures on Cayley Trees}, World Scientific Publishing Company (2013).

\bibitem{RAU} Rozikov U.A.,  Ak\i n H. and Uguz S., Exact Solution of a generalized ANNNI model on a Cayley tree,
\emph{Math. Phys. Anal. Geom.} {\bf 17}, 103-114 (2014).





%


\bibitem{V}  Vannimenus J., Modulated phase of an Ising system with competing interactions on a Cayley tree,
\textit{Zeitschrift fur Physik B Condensed Matter}, {\bf 43}(2),
141Ц148 (1981).

%
\bibitem{Zachary} Zachary S.,
Countable state space Markov random Felds and Markov chains on
trees. \textit{Ann. Prob.} \textbf{11}, 894-903 (1983).
\end{thebibliography}
\end{document}